\documentclass[pra,onecolumn]{revtex4}
\usepackage{amsmath,amsfonts,amsthm,amssymb,xspace}
\usepackage{mathtools}
\usepackage{graphicx}
\usepackage{xcolor}
\usepackage{setspace,caption}

\newcommand\mbR{\mbox{$\mathbb{R}$}}
\newcommand\mbC{\mbox{$\mathbb{C}$}}

\newcommand\ket[1]{| #1 \rangle}

\newcommand\rank{\mbox{\tt {rank}}\xspace}

\newcommand\prank{\mbox{\tt {rank}$_{\tt psd}$}\xspace}
\newcommand\kprank{\mbox{\tt {rank}$_{\tt psd}^{(k)}$}\xspace}
\newcommand\tr{\mbox{\tt {tr}}\xspace}

\newcommand\alice{\mbox{\sf Alice}\xspace}
\newcommand\bob{\mbox{\sf Bob}\xspace}
\newcommand\rcomm{\mbox{\sf {RComm}}\xspace}
\newcommand\qcomm{\mbox{\sf {QComm}}\xspace}

\newcommand\size{\mbox{\sf {size}}\xspace}
\newcommand\Q{\mbox{\sf {Q}}\xspace}
\newcommand\R{\mbox{\sf {R}}\xspace}

\newcounter{thm}
\theoremstyle{definition}
\newtheorem{definition}[thm]{Definition}

\theoremstyle{plain}
\newtheorem{prop}[thm]{Proposition}
\newtheorem{theorem}[thm]{Theorem}
\newtheorem{lemma}[thm]{Lemma}
\newtheorem{fact}[thm]{Fact}
\newtheorem{corollary}[thm]{Corollary}

\newcommand{\comment}[1]{{}}

\begin{document}

    \title{\vspace{-1cm} Quantum and Classical Hybrid Generations for Classical Correlations}

    \author{Xiaodie Lin$^{1}$, Zhaohui Wei$^{1}$, and Penghui Yao$^{2}$}
    \affiliation{$^{1}$Center for Quantum Information, Institute for Interdisciplinary Information Sciences, Tsinghua University, Beijing 100084, P. R. China\\$^{2}$State Key Laboratory for Novel Software Technology, Nanjing University, Nanjing, Jiangsu Province 210023, P. R. China}

    \begin{abstract}
     We consider two-stage hybrid protocols that combine {quantum resource and classical resource} to generate classical correlations shared by two separated players. Our motivation is twofold. First, in the near future the scale of quantum information processing is quite limited, and when {quantum resource available is not sufficient for certain tasks}, a possible way to strengthen the capability of quantum schemes is introducing extra classical resource. We analyze the mathematical structures of these hybrid protocols, and {characterize the relation between the amount of quantum resource and classical resource needed}. Second, a fundamental open problem in communication complexity theory is to describe the advantages of sharing prior quantum entanglement over sharing prior randomness, which is still widely open. {It turns out that our quantum and classical} hybrid protocols provide new insight into this important problem.
    \end{abstract}

    \maketitle

    \section{Introduction}
    Suppose two separated parties, Alice and Bob, aim to outputting random variables $X$ and $Y$, such that $(X,Y)$ is distributed exactly according to a target joint probability distribution $P$. That is to say, Alice and Bob want to sample a shared randomness $P$, and sometimes we call it a \emph{classical correlation}. {Then an important problem is, what is the minimum cost of generating an arbitrary classical correlation?}

    {Actually this problem has been systematically studied \cite{zhang2012quantum,jain2013efficient,jain2017multipartite}.} Generally, $P$ is not a product distribution, thus \alice and \bob can share a seed correlation $(X',Y')$ and each applies a local operation on the corresponding subsystem without communication. The minimum {\emph{size}} of this seed distribution, {i.e., the half of the total number of bits}, is defined to be the \emph{randomized correlation complexity} of $P$, denoted $\R(P)$. Alternatively, the two parties can also share a \emph{quantum} state $\sigma$ as a seed state, on which the two parities apply local quantum operations without communication to generate $(X,Y)$. In this case, the minimum size of the quantum seed state $\sigma$, {i.e., the half of the total number of qubits}, is called the \emph{quantum correlation complexity}, denoted $\Q(P)$.

    Instead of sharing seed states, \alice and \bob can also generate a correlation from scratch by communication only. When communicating quantum information, the minimum number of qubits exchanged between \alice and \bob, initially sharing nothing, to produce $P$ at the end of the protocol is defined as the {\em quantum communication complexity} of $P$, denoted $\qcomm(P)$. Similarly, one can also define the {\em  randomized communication complexity} of $P$, denoted $\rcomm(P)$, as the minimum number of bits exchanged to produce $P$. It turns out that for any $P$, the correlation complexity and the communication complexity are always the same, namely $\qcomm(P) = \Q(P)$ and $\rcomm(P) = \R(P)$ \cite{zhang2012quantum}. Therefore, we can simply use the notations \Q and \R to denote the quantities in quantum and classical settings respectively. In this paper, {when generating classical correlations} by quantum procedures, we will mainly focus on the setting with seed states.

    In fact, the full characterizations for \Q and \R have been achieved \cite{zhang2012quantum,jain2013efficient}, and for any classical correlation $P$. That is for any classical correlation $P$
    \begin{equation}
    \R(P) = \lceil \log_2 \rank_+(P) \rceil,
    \label{eq:qcorr+rank}
    \end{equation}
    and
    \begin{equation}
    \Q(P) = \lceil \log_2 \prank(P) \rceil.
    \label{eq:qcorrprank}
    \end{equation}
    Here for any nonnegative matrix $P\in \mbR_+^{n\times m}$, $\rank_+(P)$ is the nonnegative rank, which is defined as the minimum number $r$ such that $P$ can be decomposed as the summation of $r$ nonnegative matrices of rank $1$. And $\prank(P)$ is the positive semi-definite rank ({PSD-rank}), which is the minimum $r$ such that there are $r \times r$ positive semi-definite matrices $C_x$, $D_y\in\mbC^{r\times r}$, satisfying that $P(x,y) = \tr (C_x  D_y)$, for all $x$ and $y$ \cite{fiorini2012linear,fawzi2015positive}.

    It can be shown that the gap between nonnegative ranks and {PSD-ranks} can be huge, and this therefore reveals the remarkable advantages of quantum schemes in generating classical correlations. For example, consider the following $2^n \times 2^n$ matrix $M \in \mbR^{2^n\times 2^n}_+$ with rows and columns indexed by $n$-bit strings $a$ and $b$, and real nonnegative entries $M_{ab} := (1 - a^{\intercal} b)^2$, {where $a^{\intercal} b$ is the mod $2$ inner product between $a$ and $b$}. Then we have the following conclusions.
    \begin{fact}[\cite{fiorini2012linear}]
    It holds that $\rank_+(M) = 2^{\Omega(n)}$ and $\prank(M) = O(n)$.
    \end{fact}

    Though quantum advantages can be huge, and {extraordinary progress has been achieved} on physical implementation of quantum computation, it is widely believed that the availability of large scale quantum computers is still far \cite{arute2019quantum,preskill2018quantum}. As a consequence, in the near future the scale of quantum information processing, especially the scale of entanglement, is quite limited, say dozens or hundreds of qubits. Therefore, for some realistic classical correlations $P$, it is possible that $\lceil \log_2 \prank(P) \rceil$, {the necessary size of a shared seed quantum state that produces $P$ according to \cite{jain2013efficient}} exceeds the size that we can physically realize. In this situation, a natural question is, can we design a proper {quantum and classical hybrid protocol} to generate $P$ in such a way that, {it not only fulfills the task completely, but also fully exploits the potential of our quantum capability?} In this manuscript, by looking into the rich mathematical structures of {quantum and classical hybrid protocols}, we will give a positive answer to the above question.

    {Particularly, we first consider the case that the only restriction on our capability to manipulate quantum states is the scale, which means we can require any quantum states whenever we want as long as their size is within our means, {which may depend on the classical messages exchanged}. Then we prove that if a hybrid protocol has to be utilized to generate a large classical correlation $P$, the protocol can be fully characterized by a concept called {\emph{$k$-block positive semi-definite ranks}, which is essentially a generalization of the concept of PSD-ranks, and reveals the relation between the amount of classic resource needed and the quantum scale available. By looking into the rich mathematical structures of this new concept, we prove that the shortage of one single qubit may require a huge amount of classical resource to compensate, thus providing new evidences of quantum advantages in generating classical correlations. Furthermore, we also consider another setting with more rigorous restrictions on our freedom of exploiting quantum {resource}, i.e., in addition to the restricted quantum scale, only one quantum state is provided for the players and it is independent of classical messages. Based on the idea of entanglement transformation, we show that the second model actually has similar power with the first one.}

    {In the meanwhile, our results are also related to a famous open problem in quantum communication complexity theory. Quantum communication complexity was introduced by Yao in~\cite{Yao93}, which investigates the advantages and limit of the communication complexity models when the players are allowed to exchanges quantum messages. Dozens of examples have been discovered that exhibit the advantages of quantumness (see~\cite{10.1145/3357713.3384243} and references therein) as well as numerous methods proving the lower bounds on quantum communication complexity have been established~\cite{10.5555/1803907}. In the model introduced by Yao, the players may share classical random strings independent of the input before exchanging messages. This is named as the Yao's model. Thanks to Newman's theorem~\cite{NEWMAN199167}, we know that the shared randomness can only save at most $O(\log n)$ bits communication, where $n$ is the length of the inputs. Cleve and Buhrman in~\cite{CB97} introduced another model where the players are allowed to preshare arbitrary bipartite quantum states, which is named as the Cleve-Buhrman model. Using quantum teleportation \cite{bennett1993teleporting}, we may assume that the players in the Cleve-Buhrman model only exchange classical messages while the communication cost increases by at most factor 2. }

    {A fundamental problem in communication complexity is how much communication can be saved if the players share entanglement. In other words, what is the largest separation between the Yao's model and the Cleve-Buhrman model? The role of entanglement in quantum computing has always been a core topic in the theory of quantum computation, which is studied in various models of computation. In particular, it has been shown in a very recent breakthrough result~\cite{JNVWY'20} that multi-prover interactive proof systems with sharing entanglement are able to decide the Halting problem, while the ones without sharing entanglement are in $\mathrm{NEXP}$~\cite{babai1991non}. However, little is known about the power of entanglement in communication complexity. Indeed, till now we do not have any nontrivial upper bound on the separation between the Yao's model and the Cleve-Burhman model. Meanwhile, we are not aware of any example that exhibiting a super-constant separation between these two models as well. In this paper, our results provide more facts on the power of entanglement in the context of generating classical correlation, which show that sharing entanglement can save the classical communication significantly, and thus hopefully shed a new light on this widely open problem.}

    \section{The hybrid protocols}
    Recall that for convenience we define the {size} of a bipartite distribution as the half of the total number of bits. Similarly, the size of a bipartite quantum state is the half of the total number of qubits. We suppose the largest bipartite quantum system we can manipulate has a size of $s$ qubits, and for convenience we call it \emph{quantum capability}. We now consider a target classical correlation $P\in \mbR_+^{n\times m}$ with $s<\lceil \log_2 \prank(P) \rceil$. Clearly, we cannot generate $P$ using a purely quantum scheme.

    Therefore, we turn to analyze the possibility that combine quantum power and classical power together. To make the hybrid protocol valuable, we hope the extra classical cost needed will be dramatically smaller than that of a pure classical protocol. In the meantime, {as we have different ways to combine quantum subprotocols and classical ones for hybrid protocols in principle, we now analyze two main possibilities as below}.

    \subsection{The classical-quantum hybrid}

    {Suppose the target classical correlation can be expressed as a linear combination of two other ones, i.e., $P=\frac{1}{2}P_1+\frac{1}{2}P_2$, where $P_1$ and $P_2$ are nonnegative matrices}. Then one can easily construct examples with $\prank(P_1)<\prank(P)$ and $\prank(P_2)<\prank(P)$, which inspires us to design the following natural hybrid protocol. Assume $P=\sum_{i\in I}p_iP_i$, where $\{p_i\}$ is a probability distribution on $i\in I$, and for any $i\in I$, $P_i\in \mbR_+^{n\times m}$ is a classical correlation with $\lceil \log_2 \prank(P_i) \rceil\leq s$, then \alice and \bob can produce a sampling of $P$ as below. They first sample a shared output $i\in I$ classically according to the probability distribution $\{p_i\}$, then one of them prepares a bipartite quantum state $\rho_i$ that can serve as a seed state to produce $P_i$ and sends half of the qubits to the other party by quantum communication, which is within the quantum capability. After that, they generate a classical correlation $P_i$ by performing local measurements on $\rho_i$ like in a purely quantum protocol. Since $\sum_{i\in I}p_iP_i=P$, overall the hybrid protocol generates exactly the target classical correlation $P$.

    Since in the first stage of the protocol \alice and \bob sample $i\in I$, we call this a \emph{classical-quantum hybrid protocol}. Here the classical cost is $c=\lceil \log_2 |I| \rceil$ bits, and the quantum cost is $q=\max_i\size(\rho_i)$ qubits. Since it holds that $q\leq s$, the current hybrid protocol can generate the target correlation $P$ within the quantum capability. Below is a simple example that demonstrates this idea.

    Let
    \begin{equation}\label{eq:diagonal}
    P = \frac{1}{2^k}\begin{bmatrix}
    P_1 &  &  &  & \\
    & P_2  &  &  &\\
    &  & \ddots  &  &\\
    &  &  &  & P_{2^k}
    \end{bmatrix},
    \end{equation}
    where $2^k\cdot P\in \mbR_+^{2^kn\times 2^km}$ is a block diagonal matrix, and for the convenience of later discussion, we denote it by $\text{diag}(P_1,P_2,...,P_{2^k})$. For each $i\in[2^k]$, suppose $P_i\in \mbR_+^{n\times m}$ is a classical correlation satisfying $\prank(P_i) =2^s$. Then it can be seen that $P$, as a classical correlation, cannot be produced using a purely quantum protocol, as the current quantum capability $s$ is smaller than $\lceil \log_2 \prank(P) \rceil=k+s$. However, we can generate it using a hybrid protocol, where in the first stage it takes them classical communication of $k$ bits to sample $i\in[2^k]$, then they consume a shared quantum state of size $s$ qubits to generate the corresponding $P_i$. As long as they adjust the output labels properly, the overall output will be exactly a sample of $P$.

    As pointed out before, examples of $P_i$ exist such that $\rank_+(P_i)\gg2^s$, {i.e., when sampling $P_i$ quantum schemes enjoy remarkable advantages over classical ones}. If this is the case, though we cannot produce $P$ using a purely quantum scheme directly, such a hybrid protocol may decrease the amount of classical resource dramatically.

    Due to the above example, we are tempted to consider the following realistic problem. Still assume our quantum capability is known to be $s$, and the target classical correlation $P$ satisfies $s<\lceil \log_2 \prank(P) \rceil$. Then if we choose to generate $P$ using a classical-quantum hybrid protocol, what is the least amount of {extra classical resource} needed? Or, to put it another way, given an arbitrary classical correlation $P$, what is the minimum number $m$ such that $P$ can be expressed as the summation of $m$ nonnegative matrices with PSD-rank not larger than $2^s$? To answer this question, we first introduce the following definition, which is a generalization of the concept of PSD-rank.

    \begin{definition} A \emph{$k$-block positive semi-definite factorization} of a nonnegative matrix $P\in \mbR_+^{n\times m}$ is a collection of positive semi-definite matrices $C_i=\text{diag}(C_i^1,...,C_i^r),D_j=\text{diag}(D_j^1,...,D_j^r)\in \mbC^{kr\times kr}$ that satisfy
    \[P_{ij}=\tr (C_iD_j)=\sum_{l=1}^r\tr (C_i^lD_j^l),\ i=1,...,n,\ j=1...,m,\]
    where $C_i^l,D_j^l\in \mbC^{k\times k}$ for each $i$, $j$, and $l$. And the {$k$-block positive semi-definite rank}, denoted $\kprank(P)$, is defined as the smallest integer $r$ for which such a $k$-block positive semi-definite factorization exists.
    \end{definition}

    We now prove that the question asked above is perfectly answered by the concept of $2^s$-block semi-definite ranks, where the corresponding classical-quantum hybrid protocol is exactly characterized by an optimal $2^s$-block positive semi-definite factorization.

    \begin{theorem} Suppose the quantum capability is $s$ qubits. Then the minimum amount of classical communication needed in a classical-quantum hybrid protocol producing $P$ is exactly $\lceil\log_2{\tt rank}_{\tt psd}^{(2^s)}(P)\rceil$ bits.
    \end{theorem}
    \begin{proof} Suppose the minimal classical cost is $c$ bits. Then we have a factorization $P(x,y)=\sum_{i=1}^{2^c}p_iP_i(x,y)$, where $\{p_i\}$ is a probability distribution on $i\in [2^c]$, and each correlation $P_i$ can be generated by quantum communication of $\lceil\log_2\prank(P_i)\rceil\leq s$ qubits with a purely quantum protocol. Suppose a positive semi-definite factorization of $P_i$ is $P_i(x,y)=\tr(C_x^iD_y^i)$, where without loss of generality $C_x^i,D_y^i$ can be chosen as positive semi-definite matrices of size $2^s\times 2^s$ for any $x\in[n],y\in[m]$. Let $C_x=\text{diag}(p_1C_x^1,...,p_{2^c}C_x^{2^c})$ and $D_y=\text{diag}(D_y^1,...,D_y^{2^c})$. Then it can be seen that $C_x$ and $D_y$ are block diagonal positive semi-definite matrices with each block of size    $2^{s}\times2^{s}$, and furthermore, $P(x,y)=\tr(C_xD_y)$ for any $x\in[n],y\in[m]$. Therefore, it holds that ${\tt rank}_{\tt psd}^{(2^s)}(P)\leq2^c$, i.e., $\lceil\log_2{\tt rank}_{\tt psd}^{(2^s)}(P)\rceil\leq c$.

    On the other hand, suppose $r={\tt rank}_{\tt psd}^{(2^s)}(P)$. Then one can find block diagonal positive semi-definite matrices $C_x$ and $D_y$ of block size $2^{s}\times2^{s}$ such that $P(x,y)=\tr(C_xD_y)$ for any $x\in[n],y\in[m]$. That is to say, we can suppose $C_x=\text{diag}(C_x^1,...,C_x^{r})$ and $D_y=\text{diag}(D_y^1,...,D_y^{r})$, where $C_x^i$ and $D_y^i$ are positive semi-definite matrices of size $2^s\times 2^s$ for any $i\in[r]$. Define $P_i$ to be the classical correlation $Q_i/\Vert Q_i\Vert_1$, where $Q_i\in\mbR_+^{n\times m}$ and $Q_i(x,y)=\tr(C_x^iD_y^i)$ for any $x\in[n],y\in[m]$. Note that this is well-defined: If we let $p_i=\Vert Q_i\Vert_1$, then $p_i>0$ according to the definition of $2^s$-block diagonal positive semi-definite rank. Then it is not hard to see that $P=\sum_{i=1}^{r}p_iP_i$, and for each $i\in[r]$, it holds that $\prank(P_i)\leq2^s$. Therefore, one can design a classical-quantum hybrid protocol to generate $P$ corresponding to this factorization, where the cost of classical communication is $\lceil\log_2r\rceil$, implying that $c\leq\lceil\log_2{\tt rank}_{\tt psd}^{(2^s)}(P)\rceil$.
    \end{proof}
    In real-life implementations of sampling $P\in \mbR_+^{n\times m}$, we often allow a small deviation of $\epsilon$, which suggests us define an approximate version of $k$-block positive semi-definite rank, that is,
    \begin{equation}
    {\tt rank}_{\tt psd,\epsilon}^{(k)}(P)\equiv\min\{\kprank(Q):\text{ }Q\in \mbR_+^{n\times m}\text{ is a probability distribution and }\|P-Q\|_{1}\leq\epsilon\},
    \end{equation}
    {where $\|P-Q\|_{1}$ is the $1$-norm of $P-Q$, i.e., the summation of the absolute values of all entries of $P-Q$}. Then it can be seen that when tolerating a small additive error $\epsilon$, the cost of optimal classical-quantum protocol that samples $P$ approximately is characterized by the corresponding approximate $k$-block positive semi-definite rank.

    Therefore, we now know that given the quantum capability $s$ qubits, suppose $s<\lceil \log_2 \prank(P) \rceil$, then in order to design a proper classical-quantum hybrid protocol generating $P$, estimating ${\tt rank}_{\tt psd}^{(2^s)}(P)$ is crucial. In the rest of the current section, we will focus on the characterization of ${\tt rank}_{\tt psd}^{(2^s)}(P)$.

    Firstly, according to the properties of ranks and PSD-ranks, we immediately have the following lower bounds for ${\tt rank}_{\tt psd}^{(2^s)}(P)$.
    \begin{fact}For any nonnegative matrix $P\in \mbR_+^{n\times m}$ and any integer $k\geq1$, it holds that
    \begin{equation}\label{eq:withPSDRank}
    \kprank(P)\geq \frac{\prank(P)}{k}, \ \ {\tt rank}_{\tt psd,\epsilon}^{(k)}(P)\geq \frac{{\tt rank}_{\tt psd,\epsilon}(P)}{k},
    \end{equation}
    and
    \begin{equation}\label{eq:withRank}
    \kprank(P)\geq \frac{\rank(P)}{k^2}, \ \ {\tt rank}_{\tt psd,\epsilon}^{(k)}(P)\geq \frac{{\tt rank}_{\epsilon}(P)}{k^2},
    \end{equation}
    where ${\tt rank}_{\tt psd,\epsilon}(P)$ and ${\tt rank}_{\epsilon}(P)$ are the approximate PSD-rank and the approximate rank of $P$, respectively, i.e., ${\tt rank}_{\tt psd,\epsilon}(P)\equiv\min\{\prank(Q):\text{ }Q\in \mbR_+^{n\times m}\text{ is a probability distribution and }\|P-Q\|_{1}\leq\epsilon\}$ and ${\tt rank}_{\epsilon}(P)\equiv\min\{\rank(Q):\text{ }Q\in \mbR_+^{n\times m}\text{ is a probability distribution and }\|P-Q\|_{1}\leq\epsilon\}$.
    \end{fact}

    The above two lower bounds on exact cases are tight. For example, let $P$ be the classical correlation in Eq.\eqref{eq:diagonal}, then it holds that $\prank(P)=2^{s+k}$ and ${\tt rank}_{\tt psd}^{(2^s)}(P)\leq 2^{k}$, where the second fact comes from that we can decompose $P$ into the summation of $2^k$ classical correlations with each corresponding to one $P_i$. Hence ${\tt rank}_{\tt psd}^{(2^s)}(P)\leq \prank(P)/2^s$, and combined with Eq.\eqref{eq:withPSDRank} this means that actually ${\tt rank}_{\tt psd}^{(2^s)}(P)= \prank(P)/2^s=2^k$. Furthermore, if one chooses $P_i$ such that $\rank(P_i)=\prank(P_i)^2=2^{2s}$ for any $i\in[2^k]$, then we have $\rank(P)=2^{2s+c}$, and ${\tt rank}_{\tt psd}^{(2^s)}(P)=\rank(P)/2^{2s}$, implying that Eq.\eqref{eq:withRank} can also be tight. However, later we will see that in some cases these relations can be very loose.

    We next turn to upper bounds for $\kprank(P)$. It turns out that $\kprank(P)$ can be upper bounded by generalizing the idea in the example of Eq.\eqref{eq:diagonal}, and the notion of \emph{combinatorial rectangle} proposed by Yao \cite{yao1979some}, which plays a key role in communication complexity theory. Suppose $X\subseteq[n]$ and $Y\subseteq[m]$, then $X\times Y$ pins down a submatrix of $P$, called a combinatorial rectangle. Then we define a \emph{partition} of $P$ to be a series of nonzero combinatorial rectangles, where there is no overlap between any two of them and the union of all combinatorial rectangles contains all nonzero entries of $P$. If each combinatorial rectangle, regarded as a classical correlation after normalization, can be produced quantumly within the quantum capability, then $P$ can be generated by a classical-quantum protocol as a probability mixture of these combinatorial rectangles. Naturally, in this situation we are interested in the size of the optimal partition of $P$, which has the minimum number of combinatorial rectangles with each within the quantum capability. For this, we make the following definition.
    \begin{definition} Let $P\in \mbR_+^{n\times m}$ be a classical correlation. Define the \textbf{$k$-partition number} of $P$, denoted $C^k(P)$, as the minimum size of a partition of $P$ with the property that each combinatorial rectangle has PSD-rank at most $k$. For convenience, we call these combinatorial rectangles a \textbf{$k$-partition} of $P$.
    \end{definition}
    Then we have the following proposition.
    \begin{prop} For any nonnegative matrix $P\in \mbR_+^{n\times m}$ and any integer $k\geq1$, it holds that
    \begin{equation}
    \kprank(P)\leq C^k(p).
    \end{equation}
    \end{prop}
    \begin{proof}Suppose $t=C^k(P)$, and $\{P_1,P_2,...,P_t\}$ is an optimal $k$-\textbf{partition} of $P$. 
    Define the weight of the $i$-th combinatorial rectangle is the summation of all its entries, denoted $w_i$. Then $\sum_{i=1}^tw_i=1$, and $\{w_i,i\in[t]\}$ is a valid probability distribution.

    We expand the size of each $P_i$ to be $n\times m$ by adding zero entries with the positions of all nonzero entries the same as in $P$, which does not change its PSD-rank. For any $i\in[t]$ suppose an optimal positive semi-definite factorization of $P_i$ is $P_i(x,y)=\tr(C_x^iD_y^i)$, where $C_x^i,D_y^i$ are $k\times k$ positive semi-definite matrices for any $x\in[n],y\in[m]$. Let $C_x=\text{diag}(w_1C_x^1,...,w_{t}C_x^{t})$ and $D_y=\text{diag}(D_y^1,...,D_y^{t})$. Then it can be seen that $P(x,y)=\tr(C_xD_y)$ for any $x\in[n],y\in[m]$. Therefore, it holds that $\kprank(P)\leq C^k(p)$.
    \end{proof}
    We now consider a specific example of this upper bound. Again we go back to the one in Eq.\eqref{eq:diagonal}, and we already know that in this case ${\tt rank}_{\tt psd}^{(2^s)}(P)\leq 2^k$. Inspired by this example, the above upper bound naturally gives the same result, which means the amount of classical resource needed to perform a classical-quantum hybrid generating $P$ is at most $k$ bits. In the meantime, note that $\prank(P)=2^{s+k}$, that is to say, a purely quantum scheme producing $P$ needs a shared quantum state of size $s+k$ qubits. Therefore, it can be said that the $k$-bit classical resource involved in the classical-quantum protocol works quite efficiently, in the sense that it fulfills completely the task of the extra $k$-qubit quantum resource in a purely quantum scheme.

    However, this is not always the case: It is possible that the effect of quantum resource of one single qubit needs a large amount of classical resource to compensate! 
    Before exhibiting such an example, we would like to remark that this can be regarded as another angle to reveal the remarkable advantages of quantum resource over classical resource in generating correlations. Our example will be based on {\emph{Euclidean distance matrices}} that have been extensively studied~\cite{lin2010nonnegative,hrubevs2012nonnegative,shitov2019euclidean}.

    \begin{definition}(Euclidean Distance Matrix) Given $n$ distinct real numbers $c_1,...,c_n$, the corresponding Euclidean distance matrix (EDM) is the $n\times n$ symmetric and nonnegative matrix $Q(c_1,...,c_n)$ whose $(i,j)$-th entry $q_{i,j}$ is defined by
    \[q_{ij}=(c_i-c_j)^2,\ i,j=1,...,n.\]
    \end{definition}

    \begin{fact}\cite{shitov2019euclidean}
    There exist $n$ distinct real numbers $c_1,...,c_n$ such that $\rank(Q_1)=3,\prank(Q_1)=2$ and $\rank_+(Q_1)\geq 2\sqrt{n}-2$, where $Q_1=Q(c_1,...,c_n)$.
    \end{fact}
    We choose such a $Q_1$ with $q_{i,j}>0$ for any $i\neq j$, and let $\tilde{Q}_1=Q_1/\Vert Q_1\Vert_1$, then $\tilde{Q}_1$ is a classical correlatio,n with $\rank_+(\tilde{Q}_1)\geq2\sqrt{n}-2$. The above fact indicates that when generating $\tilde{Q}_1$, a quantum scheme enjoys remarkable advantages over any classical ones, as the cost of the former can be only one single qubit, while the latter needs at least classical resource of $\lceil\log n\rceil$ bits.

    We now consider $\tilde{Q}_2=\tilde{Q}_1\otimes\tilde{Q}_1$, which is a classical correlation of size $n^2\times n^2$, and similarly for any positive integer $k$, we define $\tilde{Q}_k=\tilde{Q}_1^{\otimes k}$. Since $\prank(A\otimes B)\leq \prank(A)\cdot\prank(B)$ for any nonnegative matrices $A$ and $,B$, we have that $\prank(\tilde{Q}_2)\leq 4$ (actually it is not hard to see that $\prank(\tilde{Q}_2)= 4$), thus a purely quantum scheme only needs a quantum seed of size 2 qubits to generate $\tilde{Q}_2$. To study classical-quantum hybrid protocols generating $\tilde{Q}_2$, we now assume that $s=1$, i.e., our quantum capability is only one qubit, thus we cannot generate $\tilde{Q}_2$ using a purely quantum scheme directly. Then we turn to classical-quantum hybrid protocols to produce $\tilde{Q}_2$, and we are interested in the minimum classical resource needed. According to Theorem 1, we have to estimate $\lceil\log_2{{\tt rank}_{\tt psd}^{(2)}}(\tilde{Q}_2)\rceil$. We now prove the following conclusion.
    \begin{prop}\label{ex1}
    ${{\tt rank}_{\tt psd}^{(2)}}(\tilde{Q}_2)\geq \log n$.
    \end{prop}
    \begin{proof}
    Denote the $(i,j)$-entry of $\tilde{Q}_1$ by $\tilde{q}_{i,j}$, i.e., $\tilde{q}_{i,j}=\tilde{Q}_1(i,j)$. Then
     \begin{equation}\label{eq:blcok}
    \tilde{Q}_2=\tilde{Q}_1\otimes\tilde{Q}_1=\begin{bmatrix}
    0 & \tilde{q}_{1,2}\tilde{Q}_1 & \dots & \tilde{q}_{1,n}\tilde{Q}_1\\
    \tilde{q}_{2,1}\tilde{Q}_1 & 0 & \dots & \tilde{q}_{2,n}\tilde{Q}_1\\
    \vdots & \vdots & \ddots & \vdots \\
    \tilde{q}_{n,1}\tilde{Q}_1 & \tilde{q}_{n,2}\tilde{Q}_1 & \dots & 0
    \end{bmatrix}.
    \end{equation}
    For the convenience of later discussion, we call $\tilde{q}_{i,j}\tilde{Q}_1$ the $(i,j)$-th block of $\tilde{Q}_2$ when $i\neq j$, and apparently for any $i\in[n]$ the $(i,i)$-th block is a zero matrix. For any other matrix $M$ with the same size $n^2\times n^2$, we also use this term to address the corresponding submatrix of $M$ with exactly the same position. Suppose $\tilde{Q}_2=\sum_{k=1}^rP_k$, where $P_k$ is a nonnegative matrix and {$\prank(P_k)\le 2$} for any $k\in[r]$. Then we need to prove that $r\geq \log n$.

    Suppose $r<\log n$. Then we claim that for any $i\neq j$, there must be an integer $k_0\in[r]$ such that the $(i,j)$-th block of $P_{k_0}$ has rank $3$ or $4$. This can be proved as below. Suppose this is not the case, i.e., for any $k\in[r]$, the rank of the $(i,j)$-th block of $P_{k}$ is $1$ or $2$, then according to the fact that for any rank-$2$ nonnegative matrix $A$ it holds that $\rank_+(A)=2$~\cite{cohen1993nonnegative}, the summation of the $(i,j)$-th blocks of all $P_{k}$ has a nonnegative rank smaller than $2\log n$. However, this summation is actually $\tilde{q}_{i,j}\tilde{Q}_1$, whose nonnegative rank is at least $2\sqrt{n}-2$, much larger than $2\log n$, which is a contradiction. Therefore, for any block, there exists $k\in[r]$ such that this block of $P_k$ has rank $3$ or $4$.

    We now fix an arbitrary $k\in[r]$, and focus on {the blocks of $P_k$ which have rank $3$ or $4$}. We claim that all these blocks can be covered by a \emph{position rectangle}, which will be explained later. Suppose the $(i,j)$-th and the $(i',j')$-th blocks, denoted $P_k^{(i,j)}$ and $P_k^{(i',j')}$, have rank $3$ or $4$, then they have PSD-rank 2, where we use the fact that $\prank(P_k)=2$ and the relation $\prank(A)\geq\sqrt{\rank(A)}$ for any nonnegative matrix $A$ \cite{gouveia2013lifts}. Note that it holds that
    \begin{equation}\label{eq:2+2}
    \prank\begin{pmatrix}\begin{bmatrix}
    P_k^{(i,j)} & *\\
    0 & P_k^{(i',j')}
    \end{bmatrix}\end{pmatrix}=
    \prank\begin{pmatrix}\begin{bmatrix}
    P_k^{(i,j)} & 0\\
    * & P_k^{(i',j')}\\
    \end{bmatrix}\end{pmatrix}=4,
    \end{equation}
    where the star can be any $n\times n$ nonnegative matrix. Since $\prank(P_k)=2$, this means that the locations of the blocks of $P_k$ with rank $3$ or $4$ have to be well-organized, and the patterns in Eq.\eqref{eq:2+2} cannot exist. Let $A = \{i\in[n]:\exists j\in[n] \text{ such that the }(i,j)\text{-th block has rank 3 or 4}\}$, and $B = \{j\in[n]:\exists i\in[n] \text{ such that the }(i,j)\text{-th block has rank 3 or 4}\}$. Then the observation given by Eq.\eqref{eq:2+2} implies that $A\cap B=\emptyset$. Therefore, if we can call the set $A\times B$ a {position rectangle}, then it covers all the positions of the blocks of $P_k$ with rank $3$ or $4$, and note also that the position rectangle does not contain any diagonal blocks.

    We now consider the corresponding position rectangles for all $P_k$. It can be seen that these rectangles may have overlap, but they need to cover all the off-diagonal blocks of $\tilde{Q}_2$, because of the fact that for each off-diagonal block there exists a $k_0\in[r]$ such that the corresponding block of $P_{k_0}$ has rank $3$ or $4$. Therefore, $r$ should be at least the minimum number of monochromatic-1 rectangles needed to cover all the 1s in the communication matrix of the inequality function, which means $r\geq \log n$~\cite{Nisan97}. This is contradicted to the assumption $r<\log n$. This completes the proof.
    \end{proof}
    {Therefore, to compensate the single-qubit shortage of quantum resource in generating $\tilde{Q}_2$, one has to consume classical resource of $\log\log n$ bits roughly, even with the quantum capability of one qubit. Note that here $n$ could be any positive integer, making a sharp difference from the example in Eq.\eqref{eq:diagonal}.}

    In fact, using the similar technique, we can strengthen the conclusion in the following two different ways. These facts on $\tilde{Q}_2$ clearly reveals the rich mathematical structure of classical-quantum hybrid protocols and $k$-block positive semi-definite rank. {Indeed, the first corollary below shows that when $n$ is large, even if the quantum capability is qutrit, i.e., only one dimension smaller than 2 qubits, any classical-quantum hybrid protocol that produces $\tilde{Q}_2$ will need a large amount of classical resource.}
    \begin{corollary}\label{corollary1}
    ${{\tt rank}_{\tt psd}^{(3)}}(\tilde{Q}_2)\geq \log n$.
    \end{corollary}
    \begin{proof} The proof is almost the same with the previous proposition, except that now the blocks $P_k^{(i,j)}$ and $P_k^{(i',j')}$ introduced above can have PSD-rank 2 or 3, but the patterns in Eq.\eqref{eq:2+2} still cannot exist. Therefore, the proof still works.
    \end{proof}
    At the same time, the following corollary implies that for any positive integer $k$, there always exist classical correlations $P$ such that the cost of a purely quantum scheme to sample $P$ is $k$ qubits, but if the quantum capacity is $k-1$ qubits, i.e., a shortage of one single qubit for a purely quantum scheme, {then in any classical-quantum hybrid protocol sampling $P$ a large amount of classical resource has to be needed}.

    \begin{corollary}\label{corollary2}
    For any positive integer $k\geq2$, ${{\tt rank}_{\tt psd}^{(2^{k-1})}}(\tilde{Q}_k)\geq \log n$.
    \end{corollary}
    \begin{proof} We prove it by induction. First, according to Proposition \ref{ex1}, we know that it is true when $k=2$. We suppose it holds when $k=i_0$, i.e., ${{\tt rank}_{\tt psd}^{(2^{i_0-1})}}(\tilde{Q}_{i_0})\geq \log n$, and we now focus on ${{\tt rank}_{\tt psd}^{(2^{i_0})}}(\tilde{Q}_{i_0+1})$. Since $\tilde{Q}_{i_0+1}$ can be expressed in a similar way as Eq.\eqref{eq:blcok}, for convenience we also use the term of the $(i,j)$-th block to address the corresponding submatrix, except that now it is not $\tilde{q}_{i,j}\tilde{Q}_1$, but $\tilde{q}_{i,j}\tilde{Q}_{i_0}$. Again we suppose $\tilde{Q}_{i_0+1}=\sum_{k=1}^rP_k$, where $P_k$ is a nonnegative matrix and $\prank(P_k)\leq 2^{i_0}$ for any $k\in[r]$. And we need to prove that $r\geq \log n$.

    Suppose $r<\log n$. Then for any $i\neq j$, there must be an integer $k_0\in[r]$ such that the $(i,j)$-th block of $P_{k_0}$, denoted $P_{k_0}^{(i,j)}$, has PSD-rank larger than $2^{i_0-1}$. If this is not true, then $\sum_{k=1}^rP_{k}^{(i,j)}$, which is actually $\tilde{q}_{i,j}\tilde{Q}_{i_0}$, can be a summation of $r<\log n$ nonnegative matrices with each having PSD-rank not larger than $2^{i_0-1}$, contradicted with the assumption that ${{\tt rank}_{\tt psd}^{(2^{i_0-1})}}(\tilde{Q}_{i_0})\geq \log n$.

    Then again we fix a $k\in[r]$ and look at the blocks of $P_k$ with PSD-rank larger than $2^{i_0-1}$. By a similar observation as Eq.\eqref{eq:2+2}, we know that these special blocks of $P_k$ also appear in a similar pattern as the blocks with rank $3$ or $4$ in the case of $\tilde{Q}_2$, and their positions can also covered by a position rectangle. Therefore, a similar argument proves that we must have $r\geq\log n$.

    \end{proof}







    \subsection{The quantum-classical hybrid}

    In classical-quantum hybrid protocols, {the major restriction on exploiting quantum power is the size of available quantum states}. Within the quantum capability, we have the freedom to control and manipulate any quantum state. Particularly, when producing a classical correlation, with respect to the classical sampling result $i$ in the first stage, we are able to ask for any corresponding quantum state $\rho_i$. However, sometimes this kind of freedom is still expensive to us. For this, we now consider a new hybrid protocol with more rigorous restrictions, that is, only one quantum state independent of classical messages is available for the players, and thus the classical-quantum hybrid protocols introduced above do not work any more. Since the quantum state is fixed, we can choose its preparation as the first action, and hence call the new protocol a \emph{quantum-classical hybrid} one.


    Given one single copy of shared quantum state, say $\rho$, one may think of utilizing it in the following natural way: Based on the shared state, \alice and \bob produce a classical correlation $P'$. After sampling $x'$ and $y'$ according to $P'$, both of them make two proper local classical samplings accordingly, then give their outputs $x$ and $y$, hoping that the final output is exactly distributed corresponding to the target $P$. However, it can be argued that, this is not possible in general. Indeed, since the second stage is a classical local sampling for each party, each operation can be regarded as a special form of quantum operation. Then if the above protocol is possible, each party can merge this special quantum operation into the local quantum operation he/she performs when producing $P'$, resulting in a valid composite quantum operation. Therefore, based on the original seed quantum state of size $s$, \alice and \bob is able to generate $P$ directly, which is a contradiction.

    Due to this observation, one may wonder, with such a rigourous restriction on quantum resource available, whether quantum can make essential contributions or not in this task? It turns out the answer to this question is {again affirmative}. To explain why this is the case, we first recall two useful facts.

    First, if we choose all bipartite quantum states $\rho_i$ involved in a classical-quantum hybrid protocol to be pure, we still have the same power in generating classical correlations, even if the quantum capability is unchanged~\cite{sikora2016minimum}. Second, we also need the following well-known result by Nielsen.
    \begin{fact}\cite{nielsen1999conditions}\label{fact:nielsen}
    $\ket\Psi$ and $\ket\Phi$ are two $d\times d$ bipartite pure quantum states, and $\lambda_{\Psi}$ and $\lambda_{\Phi}$ are the vectors of their Schmidt coefficients respectively. Then $\ket\Psi$ can be transformed to $\ket\Phi$ using local operations and classical communication (LOCC) if and only if $\lambda_{\Psi}$ is majorized by $\lambda_{\Phi}$.
    \end{fact}
    Suppose $\lambda_{\Psi}=(\lambda_{\Psi,1},...,\lambda_{\Psi,d})$ and $\lambda_{\Phi}=(\lambda_{\Phi,1},...,\lambda_{\Phi,d})$ are real $d$-dimensional vectors. We say $\lambda_{\Psi}$ is majorized by $\lambda_{\Phi}$ if for any $k\in[d]$, i.e.,
    \[\sum_{i=1}^{k}\lambda_{\Psi,i}^{\downarrow}\le \sum_{i=1}^{k}\lambda_{\Phi,i}^{\downarrow},\]
    with equality holding when $k=d$, and here the $\downarrow$ indicates the descending order of the eigenvalues.

    For example, {if} \alice and \bob share $s$ Einstein-Podolsky-Rosen (EPR) pairs, i.e., a pair of qubits which are in a maximally entangled state, then as a whole bipartite pure state the corresponding vector of Schmidt coefficients is $\lambda_{s-EPR}=(2^{-s},2^{-s},...,2^{-s})$. Then, for any $2^s\times 2^s$ pure quantum state $\ket\Phi$, it is easy to check that $\lambda_{s-EPR}$ is majorized by $\lambda_{\Phi}$.

    With the above two facts, we can design a quantum-classical hybrid protocol to generate a target classical correlation $P$ as below. Suppose that an optimal classical-quantum hybrid protocol that generates $P$ corresponds to a decomposition $P=\sum_{i\in I}p_iP_i$, and for any $i\in I$, $P_i$ can be produced quantumly using a bipartite quantum state $\rho_i$ within quantum capability $s$. According to the above discussion, we can assume that all $\rho_i$ are pure. Then in a quantum-classical hybrid protocol, \alice and \bob first share $s$ EPR pairs, which is within quantum capability. Next they sample an integer $i\in I$ classically with respect to the distribution $\{p_i\}$. After obtaining shared $i$, they transform the $s$ EPR pairs into $\rho_i$ using LOCC. According to Fact \ref{fact:nielsen}, this can be fulfilled with certainty, though needs some classical communication. Then they are able to sample $P_i$ by performing local quantum operations on $\rho_i$. It is not hard to see that the overall output will be exactly a sampling of $P$, as in a classical-quantum hybrid protocol.



    It can be seen that the resource consumptions in a quantum-classical hybrid protocol are quite similar with those in the corresponding classical-quantum hybrid protocol, except some extra classical communication is needed in the part that transforms $s$ EPR pairs into $\rho_i$, {which turns out to be at most $2^s-1$ bits~\cite{nielsen1999conditions}}. Therefore, we have the following conclusion.
    \begin{prop}\label{quantum-and-classical}
    Suppose $P$ is a classical correlation with $\prank(P)>2^s$, where $s$ is the quantum capability. Then the classical communication needed in a quantum-classical hybrid protocol to sample $P$ is at most $\lceil\log_2{\tt rank}_{\tt psd}^{(2^s)}(P)\rceil+2^s-1$ bits.
    \end{prop}
    Consider the facts that for state-of-the-art technology $s$ is still quite small, and that classical communication is relatively cheap, the performance of a quantum-classical hybrid protocol is comparable with that of the corresponding classical-quantum protocol, though it suffers from more rigorous restriction to access quantum {resource}.


    \section{The advantages of shared entanglement over shared randomness in communication complexity}

    {As mentioned before, in communication complexity theory a fundamental open problem is to exhibit and prove the advantages of shared entanglement over shared randomness in computing boolean functions}. Though hybrid protocols for generating classical correlations deal with a different and simpler task, they provide us an angle to look into the advantages of shared entanglement over shared randomness in communication protocols.

    For this, we now consider and compare the following two specific settings. The mission is still sampling a classical correlation $P$. In the two settings, \alice and \bob first share two different resources of a same size respectively: one is entangled quantum state, and the other is public randomness. We set the amount of shared resources in such a way that to fulfill the task, {they may need more computational resource, which we suppose to be quantum communication}. Therefore, we can see that one of the two settings is actually a purely quantum protocol, while the other is a classical-quantum hybrid protocol. We compare the amount of quantum communication needed in the second stage. Clearly, this is a reasonable way to compare the computational power of the shared entanglement and public randomness involved in the first stage. 

    More specifically, suppose $P\in \mbR_+^{n\times m}$ is the target classical correlation. And we let the common size of the initial shared resources be $\lceil \log_2 \prank(P) \rceil$ bits or qubits. Then in the purely quantum protocol, the quantum communication needed in the second stage is zero, as the shared quantum state in the first stage is already sufficient to sample $P$. { As a result, to compare the two settings, the remaining problem is estimating how much quantum communication is needed in classical-quantum hybrid protocols. For convenience, we denote this quantity by $t$ qubits.}

    We immediately have two trivial lower and upper bounds for $t$. First, if $\prank(P)<\rank_+(P)$, which is usually the case, then $t>0$. Second, \alice and \bob can choose to throw away the shared randomness and generate $P$ from scratch in the second stage, and the corresponding cost is $\lceil \log_2 \prank(P)\rceil$ qubits. Therefore, it holds that
    \begin{equation}\label{eq:trivialbounds}
    t\leq \lceil \log_2 \prank(P)\rceil.
    \end{equation}
    Actually, we can prove the following result, which provides a nontrivial lower bound for $t$.
    \begin{lemma} In a classical-quantum hybrid protocol that generates $P\in \mbR_+^{n\times m}$, suppose the costs of the first and the second stages are $c$ bits and $s$ qubits respectively. Then it holds that
    \begin{equation}
    2s+c\geq\lceil \log_2 \rank(P)\rceil.
    \end{equation}
    \end{lemma}
    \begin{proof}According to the structures of classical-quantum hybrid protocols, we have $P=\sum_{i=1}^{2^c}P_i$, and $\prank(P_i)\leq2^s$ for any $i\in[2^c]$. Then using the relation $\prank(A)\geq\sqrt{\rank(A)}$ for any nonnegative matrix $A$, it holds that $\rank(P_i)\leq2^{2s}$. In the meanwhile, we also have that
    \begin{equation}
    \rank(P)\leq\sum_{i=1}^{2^c}\rank(P_i)\leq 2^{2s+c},
    \end{equation}
    which concludes the proof.
    \end{proof}
    Recall that in our setting we set $c$ to be $\lceil \log_2 \prank(P) \rceil$, hence the above lemma implies the following fact.
    \begin{corollary}\label{coro:lowerboundforT}
    \begin{equation}
    t\geq \frac{1}{2}\left(\lceil \log_2 \rank(P)\rceil-\lceil \log_2 \prank(P) \rceil\right).
    \end{equation}
    \end{corollary}

    {Note that there exists nontrivial nonnegative matrices $P$ such that $\prank(P)=\sqrt{\rank(P)}$~\cite{lee2017some}}. If we choose such $P$ as our target classical correlation, the result given by Corollary \ref{coro:lowerboundforT} is actually
    \begin{equation}
    t\geq \frac{1}{2}\lceil \log_2 \prank(P)\rceil.
    \end{equation}
    This indicates that the trivial upper bound in Eq.\eqref{eq:trivialbounds} can be tight up to a factor $1/2$. 

    \section{Conclusion}
    Motivated by the fact that the scale of near-term quantum computing is quite limited, in this paper we proposal two kinds of hybrid protocols that combine classical power and quantum power to generate large-scale classical correlations. By looking into the connections between these two models, we show that their performances are close, thus we can choose to focus on the more flexible one of them, i.e., the model of classical-quantum hybrid protocols. {Particularly, we show that this kind of protocols can be fully characterized by the new {concepts of $k$-block positive semi-definite rank and $k$-block positive semi-definite factorization that we proposed}. By specific examples, we show that hybrid protocols have rich mathematical structures, which, from two different viewpoints, indicate the remarkable quantum advantages in generating classical correlations. Indeed, we witness the cases where in order to compensate the shortage of single-qubit quantum resource, a large amount of classical resource has to be consumed. In the meanwhile, by comparing two specific settings with the same amount but different {kinds} of beforehand shared {resources}, we may gain a better understanding of the different power of shared entanglement and public randomness in communication complexity theory.}


    \begin{acknowledgments}
    We thank Xun Gao and Zhengfeng Ji for helpful comments. X.L. and Z.W. are supported by the National Key R\&D Program of China, Grant No. 2018YFA0306703 and the start-up funds of Tsinghua University, Grant No. 53330100118. This work has been supported in part by the Zhongguancun Haihua Institute for Frontier Information Technology. P.Y. is supported by the National Key R\&D Program of China 2018YFB1003202, National Natural Science Foundation of China (Grant No. 61972191), the Fundamental Research Funds for the Central Universities 0202/14380068, a China Youth 1000-Talent grant and Anhui Initiative in Quantum Information Technologies Grant No. AHY150100.
    \end{acknowledgments}

    \bibliographystyle{alpha}
  \bibliography{ref}

\end{document}